\documentclass[11pt]{article}
\usepackage[T1]{fontenc}
\usepackage[latin9]{inputenc}
\usepackage{geometry}
\usepackage{amsmath}
\usepackage{amsfonts}
\usepackage{amssymb}
\usepackage{amsthm}
\usepackage{comment}
\usepackage{algorithm}
\usepackage{algorithmicx}
\usepackage{algpseudocode}
\usepackage{enumerate}
\usepackage{mathtools}
\usepackage{thm-restate}
\usepackage{mathrsfs}
\usepackage{csquotes}
\usepackage{apptools}
\usepackage{graphicx}
\usepackage{verbatim}
\usepackage{microtype}
\usepackage{float}
\usepackage{chngcntr}
\usepackage{xr}
\usepackage{bm}
\usepackage{todonotes}
\usepackage[shortlabels]{enumitem}
\usepackage{tikz}
\usepackage{pdfsync}
\usepackage{nicefrac}
\usepackage[numbers]{natbib}
\usepackage{crossreftools}
\usepackage[hypertexnames=false,backref=page]{hyperref}
\usepackage[capitalize, nameinlink]{cleveref}

\makeatletter

\numberwithin{equation}{section}
\numberwithin{figure}{section}

\theoremstyle{plain}
\newtheorem{theorem}{Theorem}[section]

\newtheorem{lemma}[theorem]{Lemma}

\theoremstyle{remark}
\newtheorem{remark}[theorem]{Remark}

\theoremstyle{definition}
\newtheorem{definition}[theorem]{Definition}

\makeatother

 \newcommand{\TODO}[1]{\textbf{\color{red}[TODO: #1]}}

\algnewcommand{\LeftComment}[1]{\(\triangleright\) #1}

\global\long\def\defeq{\stackrel{\mathrm{{\scriptscriptstyle def}}}{=}}%
\global\long\def\norm#1{\left\Vert #1\right\Vert }%
\def\eps{\varepsilon}
\global\long\def\R{\mathbb{R}}%
\global\long\def\polylog{\mathrm{polylog}\;}%

\global\long\def\ot{\overline{t}}%

\global\long\def\ct{\mathcal{T}}%
\global\long\def\cS{\mathcal{S}}%

\global\long\def\sc{\mathbf{Sc}}%
\global\long\def\tsc{\widetilde{\mathbf{Sc}}}%

\global\long\def\epssc{{\epsilon_{\mproj}}}
\global\long\def\mzero{\mathbf{0}}%

\global\long\def\ma{\mathbf{A}}%
\global\long\def\mb{\mathbf{B}}%
\global\long\def\mi{\mathbf{I}}%
\global\long\def\ml{\mathbf{L}}%
\global\long\def\mw{\mathbf{W}}%
\global\long\def\mx{\mathbf{X}}%
\global\long\def\mpi{\mathbf{\Pi}}%

\global\long\def\mga{\mathbf{\Gamma}}%

\global\long\def\mproj{\mathbf{P}}%
\global\long\def\vx{\bm{x}}%
\global\long\def\vd{\bm{d}}%
\global\long\def\vf{\bm{f}}%
\global\long\def\vb{\bm{b}}%
\global\long\def\vc{\bm{c}}%
\global\long\def\vf{\bm{f}}%
\global\long\def\vs{\bm{s}}%
\global\long\def\vv{\bm{v}}%
\global\long\def\vw{\bm{w}}%
\global\long\def\vv{\bm{v}}%
\global\long\def\vl{\bm{l}}%
\global\long\def\vu{\bm{u}}%
\global\long\def\vx{\bm{x}}%

\global\long\def\ox{\overline{\vx}}%
\global\long\def\of{\overline{\vf}}%
\global\long\def\os{\overline{\vs}}%

\global\long\def\O{\widetilde{O}}%

\global\long\def\new{{(\mathrm{new})}}

\global\long\def\prev{{(\mathrm{prev})}}
\global\long\def\init{{(\mathrm{init})}}

\global\long\def\vzero{\bm{0}}%

\global\long\def\region{H}

\newcommand{\bdry}[1]{\partial #1}
\newcommand{\skel}[1]{\partial #1 \cup F_{#1}}

\global\long\def\collN{\mathcal{H}}

\newcommand{\sep}[1]{S(#1)}
\newcommand{\elim}[1]{{F_{#1}}}

\newcommand{\nfrac}[2]{\nicefrac{1}{2}}

\global\long\def\range#1{\operatorname{Range}(#1)}
\global\long\def\poly{\operatorname{poly}}

\global\long\def\tw{\mathrm{tw}}
\let\ref\cref

\begin{document}
	\title{Faster Min-Cost Flow and Approximate Tree Decomposition on Bounded Treewidth Graphs
	}
	\author{
Sally Dong\thanks{\texttt{sallyqd@uw.edu}. University of Washington.}
\and
Guanghao Ye\thanks{\texttt{ghye@mit.edu}. MIT.  Supported by NSF awards CCF-1955217 and DMS-2022448.} 
}
	\date{\today}
	\maketitle
	\begin{abstract}
		We present an algorithm for min-cost flow in graphs with $n$ vertices and $m$ edges, given a tree decomposition of width $\tau$ and size $S$, and polynomially bounded, integral edge capacities and costs, running in $\widetilde{O}(m\sqrt{\tau} + S)$ time.
		This improves upon the previous fastest algorithm in this setting achieved by the bounded-treewidth linear program solver of~\cite{gu2022faster, separableLP}, which runs in $\widetilde{O}(m \tau^{(\omega+1)/2})$ time, where $\omega \approx 2.37$ is the matrix multiplication exponent. 
		Our approach leverages recent advances in structured linear program solvers and robust interior point methods (IPM).
		In general graphs where treewidth is trivially bounded by $n$, the algorithm runs in $\widetilde{O}(m \sqrt n)$ time, which is the best-known result without using the Lee-Sidford barrier or $\ell_1$ IPM, demonstrating the surprising power of robust interior point methods.
		
		As a corollary, we obtain a $\widetilde{O}(\tw^3 \cdot m)$ time algorithm to compute a tree decomposition of width $O(\tw\cdot \log(n))$, given a graph with $m$ edges.
	\end{abstract}
	
	\newpage
	\tableofcontents
	\newpage
	
	\section{Introduction}

An active area of research in recent years is the advancement of interior point methods (IPM) for linear and convex programs, with its origins tracing back to the works of~\cite{10.1007/BF02579150,renegar1988polynomial}.
This, along with the design of problem-specific data structures supporting the IPM, this had led to breakthroughs in faster general linear program solvers~\cite{CohenLS21} and faster max flow algorithms~\cite{lee2014path,BrandLLSSSW21,GaoLP21:arxiv,BGJLLPS21,chen2022maximum}, among others.

One line of research, inspired by nested dissection from~\cite{lipton1979generalized} and methods in numerical linear algebra~\cite{davis2006direct}, focuses on exploiting any separable structure in the constraint matrix of the linear program, which can be characterized by first associated a graph with the matrix.
Given a graph $G = (V,E)$ on $n$ vertices, we say $S \subseteq V$ is a \emph{balanced vertex separator} if there exists some constant $b \in (0,1)$ such that every connected component of $G \setminus S$ has size\footnote{An alternative weighted definition is sometimes used: For any weighting of the vertices, the connected components of $G \setminus S$ each has weight at most $b \cdot W$, where $W$ is the total weight.} at most $b \cdot n$.
More specifically, we say $S$ is a $b$-balanced separator.
We say $G$ is \emph{$f(n)$-separable} if any subgraph $H$ of $G$ has a balanced separator of size $f(|V(H)|)$.
The \emph{treewidth} of a graph informally measures how close a graph is to a tree, and is closely related to separability. 
If $G$ has treewidth $\tau$, then any subgraph of $G$ has a $1/2$-balanced separator of size $\tau+1$;
conversely, if $G$ is $\tau$-separable, then $G$ has treewidth at most $\tau \log n$ (c.f.~\cite{bodlaender1995approximating}).
By leveraging properties of separable graphs, 
the sequence of three papers \cite{treeLP, planarFlow, separableLP} have iteratively refined the robust IPM framework and associated data structures for solving structured linear programs.

\cite{treeLP} gave the first general linear program solver parametrized by treewidth:
Given an LP of the form $\min \{ \vc^\top \vx : \ma \vx = \vb, \vx \geq \vzero\}$ and a width $\tau$ decomposition of the \emph{dual graph}\footnote{The dual graph of a matrix $\ma \in \R^{n \times m}$ is a graph on $n$ vertices, where each column of $\ma$ gives rise to a clique (or equivalently, a hyper-edge). Importantly, when the linear program is a flow problem on a graph, $G_{\ma}$ is precisely the graph in the original problem.}
$G_{\ma}$ of the constraint matrix $\ma$, 
suppose the feasible region has inner radius $r$ and outer radius $R$, and the costs are polynomially bounded.
Then the LP can be solved to $\eps$-accuracy in $\O(m \tau^2 \log(R/(\eps r))$ time.
The $\tau^2$ factor arises from carefully analyzing the sparsity-pattern of the \emph{Cholesky factorization} of $\ma \ma^\top$ and the associated matrix computation times.
This run-time dependence on $\tau$ was improved from quadratic to $(\omega+1)/2 \approx 3.37/2 = 1.68$ in \cite{gu2022faster}, by means of a coordinate-batching technique applied to the updates during the IPM.
\cite{gu2022faster} further combined \cite{treeLP} with \cite{zhang2018sparse} to give the current best algorithm for semi-definite programs with bounded treewidth.

The LP solver from \cite{treeLP} is now subsumed by \cite{separableLP}, which shows that under the same setup, except where $G_{\ma}$ is known to be a $\kappa n^\alpha$-separable graph\footnote{Here, $\kappa$ is a constant expression, but could be a function of the size of $G_{\ma}$ (which is considered fixed).}
the LP can be solved in $\O\left( \left( \kappa^2 (m + m^{2\alpha + 0.5}) + m \kappa^{(\omega-1)/(1-\alpha)} + \kappa^\omega n^{\alpha \omega} \right) \cdot \log (\frac{mR}{\eps r})\right)$ time.
Taking $\kappa = \tau$ and $\alpha = 0$ for bounded treewidth graphs recovers the \cite{treeLP} result. 
For $n^\alpha$-separable graphs, the run-time expression simplifies to $\O \left( (m + m^{1/2+2\alpha}) \log(\frac{mR}{\eps r}) \right)$.

In the special setting of flow problems, whereby the constraint matrix $\ma$ is the vertex-edge incidence matrix of a graph and $\ma \ma^\top$ is its Laplacian, fast Laplacian solvers~\cite{spielman2004nearly,JambulapatiS21} and approximate Schur complements~\cite{DurfeeKPRS17} can be combined with the separable structure of the graph to speed up the matrix computations further at every IPM step.
Using these ideas, \cite{planarFlow} gave a nearly-linear time min-cost flow algorithm for planar graphs, which are $O(\sqrt{n})$-separable.

In this work, we expand the landscape with an improved result for min-cost flow problems on bounded treewidth graphs. Previously, \cite{fomin2018fully} showed how to compute a max vertex-disjoint flow in $O(\tau^2 \cdot n \log n)$ time given a directed graph with unit capacities and its tree decomposition of width $\tau$; their approach is combinatorial in nature.

\begin{restatable}[Main result]{theorem}{MainThm} \label{thm:main}
	Let $G=(V,E)$ be a directed graph with $n$ vertices and $m$ edges.
	Assume that the demands $\vd$, edge capacities $\vu$ and costs $\vc$ are all integers and bounded by $M$ in absolute value. 
	Given a tree decomposition of $G$ with width $\tau$ and size $S \leq n \tau$, there is an algorithm that computes a minimum-cost flow in $\O(m\sqrt{\tau}\log M + S)$ expected time\footnote{Throughout the paper, we use $\widetilde{O}(\cdot)$ to hide $\polylog m$ and $\polylog M$ factors.}.
\end{restatable}

As a direct corollary of \cref{thm:main}, we can solve min-cost flow on any graph with $n$ vertices, $m$ edges and integral polynomially-bounded costs and constraints in $\O(m \sqrt{n})$ expected time, as the treewidth of any $n$-vertex graph is at most $\tau = n$, and the tree decomposition is trivially the graph itself.
This result matches that of \cite{lee2014path} (later improved to $\O(m + n^{1.5})$ by \cite{BrandLLSSSW21}) obtained using the Lee-Sidford barrier for the IPM, which requires $\O(\sqrt{n})$ iterations.
In contrast, we use the standard log barrier which requires $\O(\sqrt{m})$ iterations, and leverage the robustness of the IPM and custom data structures to reduce the amortized cost per iteration.

Using our faster max-flow algorithm as a subroutine, we can efficiently compute a tree decomposition of any given graph, where the width is within a $O(\log n)$-factor of the optimal:
\begin{restatable}[Approximating treewidth]{corollary}{ApproxTW} \label{cor:tw}
	Let $G=(V,E)$ be a graph with $n$ vertices and $m$ edges.
	There is an algorithm to find a tree decomposition of $G$ with width at most $O(\tw(G) \cdot\log n)$ in $\O(\tw(G)^3 \cdot m)$ expected time.
\end{restatable}

It is well known that computing the treewidth of a graph is NP-hard~\cite{arnborg1987complexity} and there is conditional hardness result for even constant-factor approximation algorithm~\cite{austrin2012inapproximability}. For polynomial time algorithms, the best known result is a $O(\sqrt{\log \tw})$-approximation algorithm by \cite{feige2008improved}. 

There is a series of works focused on computing approximate treewidth for small treewidth graph in near-linear time, we refer the readers to~\cite{bodlaender2016c} for a more detailed survey.
Notably, \cite{fomin2018fully} showed for any graph $G$, there is an algorithm to compute a tree decomposition of width $O(\tw(G)^2)$ in $\O(\tw(G)^7 \cdot n)$ time. \cite{brandt2019approximating} improved the running time to $\O(\tw(G)^3 \cdot m)$ with slightly compromised approximation ratio $O(\tw(G)^2 \cdot \log^{1+o(1)} n)$.
More recently, \cite{bernstein2022deterministic} showed how to compute a tree decomposition of width $O(\tw(G) \cdot \log^3 n)$ in $O(m^{1 + o(1)})$ time.

	\section{Overview}

	The foundation of our algorithm is the planar min-cost flow algorithm from~\cite{planarFlow}.
	We begin with the identical robust IPM algorithm in abstraction, as given in \cref{alg:IPM_centering}.
	\cite{planarFlow} first defines a separator tree $\ct$ for the input graph,
	and uses it as the basis for the data structures in the IPM.
	We modify the separator tree construction, so that instead of recursively decomposing the input planar graph which is $O(\sqrt{n})$-separable, we recursively decompose the $\tau$-separable bounded treewidth graph.
	The leaf nodes of $\ct$ partition the edges of the input graph;
	We guarantee that each leaf node contains $O(\tau)$-many edges, compared to constantly-many in the planar case.
	
	There are two main components to the data structures from~\cite{planarFlow}:
	\begin{enumerate}
		\item A data structure \textsc{DynamicSC} (\cref{thm:apxsc}) is used to maintain an approximate Laplacian and an approximate Schur complement matrix at every node of of the separator tree $\ct$, corresponding to the Laplacian of the subgraph represented by the node, and its Schur complement onto the boundary vertices.
		This data structure implicitly represents the matrix $(\mb^\top \mw \mb)^{-1}$, where $\mw$ are edge weights changing at every step of the IPM. 
		We use \textsc{DynamicSC} in exactly the same way.
		
		\item A data structure \textsc{MaintainSoln} (\cref{thm:sol-maintain}) using $\ct$ which implicitly maintains the primal and dual solutions $\vf$ and $\vs$ at the current IPM iteration, and explicitly maintains approximate solutions $\of$ and $\os$ at the current IPM iteration. 
		The approximate solutions are used in the subsequent iteration to compute the step direction $\vv$ and edge weights $\vw$.
	\end{enumerate}
	
	In \cite{planarFlow}, the approximate solutions $\of$ and $\os$ are updated coordinate-wise, whenever a coordinate is sufficiently far from the true value.
	A coordinate update induces updates in the data structures as follows:
	If $\of_e$ or $\os_e$ is updated for an edge $e$, (subsequently, $\vw_e$ and $\vv_e$ are updated), then we find the unique leaf node $H$ in $\ct$ containing the edge $e$, and must update \textsc{DynamicSC} and \textsc{MaintainSoln} at all nodes along the path from $H$ to the root of $\ct$.
	\cite{planarFlow} shows this runtime depends on the sizes of the nodes visited.
	
	\cite{gu2022faster} introduced a natural batching technique for the coordinate updates, 
	where coordinates representing edges are grouped into blocks, and coordinate-wise updates are performed block-wise instead.
	Since the leaves of our separator tree $\ct$ contain $O(\tau)$ edges, 
	it is natural for us to also incorporate a blocking scheme, where the blocks are given by the edge partition according to the leaves of $\ct$.
	In this case, the runtime for data structure updates is the same whether we update a single coordinate or a block containing said coordinate, since they affect the same path from the leaf node to the root of $\ct$.
	We bound the overall runtime expression using properties of the new separator tree.
	
   Lastly, the RIPM guarantees that for each $k$ iterations, $\O(k^2)$-many blocks of $\of$ and $\os$ need to be updated, meaning that running our data structures for many IPM iterations leads to superlinear runtime scaling.
   We can, however, restart our data structures at any point, by explicitly computing the current exact solutions $\vf, \vs$ and reinitializing the data structures with their values in $\O(m)$ total time. 
   On balance, we choose to restart our data structures every $\sqrt{m/\tau}$-many iterations, 
   for a total restarting cost of $\O(m \sqrt{\tau})$.
   Between each restart, we make $O(m/\tau)$-many block updates using $\O(\tau)$ time each, for a total update cost of $\O(m \sqrt{\tau})$. 
   Hence, the overall run-time is $\O(m\sqrt{\tau})$.

	\section{Robust interior point method\label{subsec:IPM}}

For the sake of completion, we give the robust interior point method
developed in \cite{treeLP}, which is a refinement of the methods in
\cite{CohenLS21, van2020deterministic}, for solving linear programs of the
form
\begin{equation}
	\min_{\vf\in\mathcal{F}} \vc^{\top}\vf\quad\text{where}\quad\mathcal{F}=\{\mb^{\top}\vf=\vb,\; \vl\leq\vf\leq\vu\}\label{eq:LP}
\end{equation}
for some matrix $\mb\in\mathbb{R}^{m\times n}$. 
\begin{algorithm}
	\caption{Robust Interior Point Method from \cite{treeLP}\label{alg:IPM_centering}}
	\begin{algorithmic}[1]
		
		\Procedure{$\textsc{RIPM}$}{$\mb \in \mathbb{R}^{m \times n}, \vb, \vc,\vl,\vu,\epsilon$}
		
		\State Let $L=\| \vc \|_{2}$ and $R=\|\vu-\vl\|_{2}$
		
		\State Define $\phi_{i}(x)\defeq-\log(\vu_{i}-x)-\log(x-\vl_{i})$ 
		\State Define $\mu_i^t(\vf_i,\vs_i) \defeq {\vs_i}/{t}+\nabla \phi_i(\vf_i)$

		\State Define $\gamma^t(\vf,\vs)_i \defeq \|(\nabla^2 \phi_i(\vf_i))^{-1/2} \mu_i^t(\vf_i,\vs_i)\|_2$
		
		\Statex
		
		\Statex $\triangleright$ Modify the linear program and obtain an
		initial $(\vf, \vs)$ for modified linear program
		
		\State Let $t=2^{21}m^{5}\cdot\frac{LR}{128}\cdot\frac{R}{r}$
		
		\State Compute $\vf_{c}=\arg\min_{\vl\leq\vf\leq\vu}\vc^{\top}\vf+t\phi(\vf)$
		and $\vf_{\circ}=\arg\min_{\mb^{\top}\vf=\vb}\|\vf-\vf_{c}\|_{2}$
		
		\State Let $\vf=(\vf_{c},3R+\vf_{\circ}-\vf_{c},3R)$ and $\vs=(-t  \nabla\phi(\vf_{c}),\frac{t}{3R+\vf_{\circ}-\vf_{c}},\frac{t}{3R})$
		
		\State Let the new matrix $\mb^{\mathrm{new}} \defeq [\mb;\mb;-\mb]$, the
		new barrier
		\[
		\phi_{i}^{\mathrm{new}}(x)
		=\begin{cases}
			\phi_{i}(x) & \text{if }i\in[m],\\
			-\log x & \text{else}.
		\end{cases}
		\]
		
		\Statex 
		
		\Statex $\triangleright$ Find an initial $(\vf,\vs)$ for the original
		linear program
		
		\State $((\vf^{(1)},\vf^{(2)},\vf^{(3)}),(\vs^{(1)},\vs^{(2)},\vs^{(3)}))\leftarrow\textsc{Centering}(\mb^{\mathrm{new}},\phi^{\mathrm{new}},\vf,\vs,t,LR)$
		
		\State $(\vf,\vs)\leftarrow(\vf^{(1)}+\vf^{(2)}-\vf^{(3)},\vs^{(1)})$
		
		\Statex 
		
		\Statex $\triangleright$ Optimize the original linear program
		
		\State $(\vf,\vs)\leftarrow\textsc{Centering}(\mb,\phi,\vf,\vs,LR,\frac{\epsilon}{4m})$
		
		\State \Return $\vf$
		
		\EndProcedure
		
		\Statex
		
		\Procedure{$\textsc{Centering}$}{$\mb,\phi,\vf,\vs,t_{\mathrm{start}},t_{\mathrm{end}}$}
		
		\State Let $\alpha=\frac{1}{2^{20}\lambda}$ and $\lambda=64\log(256m^{2})$
		where $m$ is the number of rows in $\mb$
		
		\State Let $t\leftarrow t_{\mathrm{start}}$, $\of\leftarrow\vf,\os\leftarrow\vs,\ot\leftarrow t$
		
		\While{$t\geq t_{\mathrm{end}}$}
		
		\State Set $t\leftarrow\max((1-\frac{\alpha}{\sqrt{m}})t,t_{\mathrm{end}})$
		
		\State Update $h=-\alpha/\|\cosh(\lambda\gamma^{\ot}(\of,\os))\|_{2}$
	\label{line:step_given_begin}
		
		\State Update the diagonal weight matrix $\mw=\nabla^{2}\phi(\of)^{-1}$\label{line:step_given_3}
		
		\State Update the direction $\vv$ where  $\vv_{i}=\sinh(\lambda\gamma^{\ot}(\of,\os)_{i})\mu^{\ot}(\of,\os)_i$\label{line:step_given_end} 

		\State Pick $\vv^{\|}$ and $\vv^{\perp}$ such that $\mw^{-1/2}\vv^{\|}\in\mathrm{Range}(\mb)$,
		$\mb^{\top}\mw^{1/2}\vv^{\perp}=\vzero$ and\label{line:step_user_begin}
		\begin{align*}
			\|\vv^{\|}-\mproj_{\vw}\vv\|_{2} & \leq\alpha\|\vv\|_{2},\\
			\|\vv^{\perp}-(\mi-\mproj_{\vw})\vv\|_{2} & \leq\alpha\|\vv\|_{2}
			\tag{$\mproj_{\vw}\defeq\mw^{1/2}\mb(\mb^{\top}\mw\mb)^{-1}\mb^{\top}\mw^{1/2}$}
		\end{align*}
		
		\State Implicitly update $\vf\leftarrow\vf+h\mw^{1/2}\vv^{\perp}$,
		$\vs\leftarrow\vs+\ot h\mw^{-1/2}\vv^{\|}$ \label{line:maintain_f_s}
		
		\State Explicitly maintain $\of,\os$ such that $\|\mw^{-1/2}(\of-\vf)\|_{\infty}\leq\alpha$
		and $\|\mw^{1/2}(\os-\vs)\|_{\infty}\leq\ot\alpha$ \label{line:step_user_end}
		
		\State Update $\ot\leftarrow t$ if $|\ot-t|\geq\alpha\ot$
		
		\EndWhile
		
		\State \Return $(\vf,\vs)$\label{line:step_user_output}
		
		\EndProcedure
		
	\end{algorithmic}
	
\end{algorithm}

\begin{theorem}[\cite{treeLP}] 
	\label{thm:IPM}
	Consider the linear program 
	\[
	\min_{\mb^{\top}\vf=\vb,\; \vl\leq\vf\leq\vu}\vc^{\top}\vf
	\]
	with $\mb\in\R^{m\times n}$. 
	Suppose there exists some interior point $\vf_{\circ}$ satisfying 
	$\mb^{\top}\vf_{\circ}=\vb$ and $\vl+r\leq \vf_{\circ} \leq \vu-r, $\footnote{For any vector $\vv$ and scalar $x$, we define $\vv+x$ to be the vector obtained by adding $x$ to each coordinate of $\vv$. We define $\vv-x$ to be the vector obtained by subtracting $x$ from each coordinate of $\vv$.} 
	for some scalar $r>0$.
	Let $L=\|\vc\|_{2}$ and $R=\|\vu-\vl\|_{2}$. For any $0<\epsilon\leq1/2$,
	the algorithm $\textsc{RIPM}$ (\cref{alg:IPM_centering}) finds $\vf$ such that $\mb^{\top}\vf=\vb$,
	$\vl\leq\vf\leq\vu$ and
	\[
	\vc^{\top}\vf\leq\min_{\mb^{\top}\vf=\vb,\; \vl\leq\vf\leq\vu}\vc^{\top}\vf+\epsilon LR.
	\]
	Furthermore, the algorithm has the following properties:
	\begin{itemize}
		\item Each call of $\textsc{Centering}$ involves $O(\sqrt{m}\log m\log(\frac{mR}{\epsilon r}))$
		many steps, and $\ot$ is only updated $O(\log m\log(\frac{mR}{\epsilon r}))$
		times.
		\item In each step of \textsc{Centering}, the coordinate $i$ in $\mw,\vv$ changes only if $\of_{i}$
		or $\os_{i}$ changes.
		\item In each step of \textsc{Centering}, $h\|\vv\|_{2}=O(\frac{1}{\log m})$.
		\item \cref{line:step_given_begin} to \cref{line:step_given_end} takes $O(K)$
		time in total, where $K$ is the total number of coordinate changes
		in $\of,\os$.
	\end{itemize}
	\qed
\end{theorem} 

We note that this algorithm only requires access to
$(\of,\os)$, but not $(\vf,\vs)$ during the main while-loop in \textsc{Centering}.
Hence, $(\vf,\vs)$ can be implicitly maintained via any data structure.  

\section{Nested dissection on bounded treewidth graphs} \label{subsec:nested-dissection}
In this section, we show how to leverage the structural properties of bounded treewidth graphs to
find a sparse Cholesky factorization of $\ml \defeq \mb^\top \mw \mb$, and hence implicitly maintain an approximation of $\ml^{-1}$ as part of the projection matrix 
$\mproj_{\vw}\defeq\mw^{1/2}\mb(\mb^{\top}\mw\mb)^{-1}\mb^{\top}\mw^{1/2}$ used in the RIPM.

\subsection{Separator tree for bounded treewidth graph}
The notion of using a \emph{separator tree} to represent the recursive decomposition of a separable graph is well-established in literature, c.f \cite{eppstein1996separator, henzinger1997faster}. In~\cite{planarFlow}, the authors show that the separator tree can be used to construct a sparse approximate
projection matrix for RIPM. Here, we extends their result to bounded treewidth graphs. 

\begin{definition}[$\tau$-separator tree]
	Let $G$ be a graph with $n$ vertices and $m$ edges. A \emph{separator tree} $\ct$ of $G$ is a rooted binary tree whose nodes represent a recursive decomposition of $G$ based on balanced vertex separators.

	Formally, each node $H$ of $\ct$ is a \emph{region} (edge-induced subgraph) of $G$; we denote this by $H \in \ct$. 
	At a node $H$, we store subsets of vertices $\bdry{H}, \sep{H}, \elim{H}$, 
	where $\bdry{\region}$ is the set of \emph{boundary vertices} of $H$, i.e.\ vertices with neighbours outside $\region$ in $G$;
	$\sep{\region}$ is a balanced vertex separator of $\region$;
	and $\elim{\region}$ is the set of \emph{eliminated vertices} at $\region$. 

	In a $\tau$-separator tree, the nodes and associated vertex sets are defined recursively in a top-down manner as follows: 
	
	\begin{enumerate}
		\item The root of $\ct$ is the node $\region = G$, with $\bdry{\region} = \emptyset$ and $\elim{\region} = \sep{\region}$.
		\item A non-leaf node $\region \in \ct$ has exactly two children $H_1, H_2 \in \ct$ that form a edge-disjoint partition of $\region$.
		If $|V(H)| \geq \Omega(\tau)$, then the intersection of the vertex sets $V(H_1) \cap V(H_2)$ is a balanced vertex separator $\sep{\region}$ of $\region$, with $|S(H)|\leq \tau$.
		Define the set of eliminated vertices at $H$ to be $\elim{\region} \defeq \sep{\region} \setminus \bdry{\region}$.
		
		By definition of boundary vertices, we have $\bdry{H_1} \defeq (\bdry{\region} \cup \sep{\region}) \cap V(H_1)$, and  
		$\bdry{H_2} \defeq (\bdry{\region} \cup \sep{\region}) \cap V(H_2)$.
		
		\item If $H$ is a region with $|E(H)| \leq \Theta(\tau)$, then we stop the recursion and $\region$ becomes a leaf node. 
		Define $\sep{\region} = \emptyset$ and $\elim{\region} = V(\region) \setminus \bdry{\region}$. 
	\end{enumerate}
	
	By construction, the leaf nodes of $\ct$ partition the edges of $G$.
	If $H$ is a leaf node, let $E(H)$ denote the edges contained in $H$.
	If $H$ is not a leaf node, let $E(H)$ denote the union of all the edges in the leaf nodes in the subtree $\ct_H$.
	Let $\eta$ denote the height of $\ct$.
\end{definition}

Next, we show how to construct an appropriate separator tree for bounded treewidth graphs. 

\begin{theorem}\label{thm:construct-separator-tree}
	Let $G$ be a graph with $n$ vertices and $m$ edges. Given a tree decomposition of $G$ with width $\tau$, we can construct a $\tau$-separator tree $\ct$ of $G$ with height $\eta=O(\log m)$ in $\O(n\tau)$ time.
\end{theorem}
\begin{remark}\label{rmk:boundary-size}
	Given a $\tau$-separator tree with $\eta = O(\log m)$, then $|F_H\cup \partial H|<O(\tau\poly\log(m))$.
\end{remark}

Before prove the theorem above, we need the following lemma about balanced vertex separators.
\begin{lemma}[{\cite[Theorem 4.17]{treeLP}}]
	Let $(X, T)$ be a width-$\tau$ tree decomposition of a graph $G$ on $n$ vertices. Then in $O(n \tau)$ time, we can find a 2/3-balanced vertex separator $\left(A, S, B\right)$ of $G$, and tree decompositions $\left(X_1, T_1\right)$ of $G[A \cup S]$ and $\left(X_2, T_2\right)$ of $G[B \cup S]$ each of width at most $\tau$.
\end{lemma}

\begin{proof}[Proof of \cref{thm:construct-separator-tree}]
	We note that for any subgraph $H$ of $G$, using the lemma above, we can find the $2/3$-balanced vertex separator in time $O(|V(H)|\cdot \tau)$. Then, we construct the separator tree $\ct$ recursively as follows, starting with the subgraph $H = G$:
	\begin{enumerate}
		\item Given a subgraph $H$ of $G$, if $|E(H)| \leq \Theta(\tau)$, then we stop the recursion and $H$ becomes a leaf node.
		\item If$|V(H)|> 2\tau$, 
		then we can find a 2/3-balanced vertex separator $\left(A, S, B\right)$ of $H$ in time $O(|V(H)|\cdot \tau)$. 
		Then let $H_1$ have vertex set $A \cup S$ and contain all edges incident to $A$, and let $H_2$ have vertex set $B \cup S$ and contain all edges incident to $B$. 
		We partition the edges in $\sep\region$ arbitrarily into $H_1$ and $H_2$.
		\item If $|V(H)| \leq 2\tau$, then we partition the edges of $H$ into two sets each of size at most $\frac23 |E(H)|$, and let $H_1$ and $H_2$ be graphs on $V(H)$ with their respective edge sets.\footnote{We use $2\tau$ here to differentiate between the second and third case, instead of the more natural $\tau$, in order to avoid any infinite-loop edge cases in the recursive process arising from division and rounding.}
	\end{enumerate}
	Consider a non-leaf node $H$ with children $H_1$ and $H_2$, we note that \[
		|V(H_i)|\cdot |E(H_i)| \leq \frac{2}{3}|V(H)|\cdot |E(H)| \quad \text{for } i\in\{1,2\}.
	\]
	This directly shows the height of the separator tree $\ct$ is $O(\log m)$.
	
	The running time directly follows by the lemma above and the fact that we can find a balanced edge partition in $O(|E(H)|)$ time for any subgraph $H$ of $G$.
\end{proof}

\subsection{Nested dissection using a separator tree}
\begin{definition}[Block Cholesky decomposition]
	The \emph{block Cholesky decomposition} of a symmetric matrix $\ml$ with two blocks indexed by $F$ and $C$ is:
	\begin{equation*}\label{eq:basic_chol}
		\ml
		= 
		\left[
		\begin{array}{cc}
			\mi
			& \mzero \\
			\ml_{C,F} (\ml_{F,F})^{-1}
			& \mi
		\end{array}
		\right] 
		\left[
		\begin{array}{cc}
			\ml_{F,F}
			& \mzero \\
			\mzero
			& \sc(\ml, C)
		\end{array}
		\right] 
		\left[
		\begin{array}{cc}
			\mi
			& (\ml_{F,F})^{-1}\ml_{F,C} \\
			\mzero & \mi
		\end{array}\right],
	\end{equation*}
where the middle matrix in the decomposition is a block-diagonal matrix with blocks indexed by $F$ and $C$, with the lower-right block being the \emph{Schur complement} $\sc(\ml,C)$ of $\ml$ onto $C$:
\[
\sc(\ml,C) \defeq \ml_{C,C} - \ml_{C,F}\ml_{F,F}^{-1}\ml_{F,C}. 
\]
\end{definition}

We use the separator tree structure to factor the matrix $\ml^{-1} \defeq (\mb^\top \mw \mb)^{-1}$:
\begin{theorem}[Approximate $\ml^{-1}$ factorization, c.f. \cite{planarFlow}] \label{thm:L-inv-factorization}
	Let $\ct$ be the separator tree of $G$ with height $\eta$.
	For each node $H \in \ct$ with edges $E(H)$, let $\mb[H] \in \R^{n \times m}$ denote the matrix $\mb$ restricted to rows indexed by $E(H)$,
	and define $\ml[H] \defeq \mb[H]^\top \mw {\mb[H]}$.
	
	Given approximation parameter $\epssc$, suppose for each node $H$ at level $i$ of $\ct$, we have a matrix $\ml^{(H)}$ satisfying the $e^{i \epssc}$-spectral approximation
	\begin{align}
		\ml^{(\region)} &\approx_{i\epssc} \sc(\ml[\region], \bdry{\region}\cup \elim{\region}).
	\end{align}
	Then, we can approximate $\ml^{-1}$ by
	\begin{equation} \label{eq:L-inv-factorization}
		\ml^{-1} \approx_{\eta\epssc}
		\mpi^{(0)\top}\cdots\mpi^{(\eta-1)\top} \mga
		\mpi^{(\eta-1)}\cdots\mpi^{(0)},
	\end{equation} 
	where
	\begin{equation}
	\mga \defeq 
	\left[
	\begin{array}{cccc}
		\sum_{H \in \ct(0)} \left(\ml^{(H)}_{F_H, F_H}\right)^{-1} & \mzero & \mzero\\
		\mzero &  \ddots & \mzero\\
		\mzero & \mzero & \sum_{H \in \ct(\eta)} \left(\ml^{(H)}_{F_H, F_H}\right)^{-1}
	\end{array}
	\right],
	\end{equation}
	and for $i = 0, \dots, \eta-1$,
	\begin{equation}
		\mpi^{(i)} \defeq \mi - \sum_{H \in \ct(i)} \mx^{(H)},
	\end{equation}
	where $\ct(i)$ is the set of nodes at level $i$ in $\ct$, the matrix $\mpi^{(i)}$ is supported on $\bigcup_{H \in \ct(i)} \skel{H}$ and padded with zeros to full dimension,
	and for each $H \in \ct$,
	\begin{equation}\label{def:mx^(H)}
		\mx^{(H)} \defeq \ml^{(H)}_{\bdry{\region}, F_H} \left( \ml^{(H)}_{F_H, F_H}\right)^{-1}.
	\end{equation}
\end{theorem}

\cite{planarFlow} gives a data structure to maintain an implicit representation of $\ml^{-1}$ as the weights $\vw$ undergoes changes in the IPM:

\begin{theorem}[{\cite[Theorem 6]{planarFlow}}]
	\label{thm:apxsc}
	Given a graph $G$ with $m$ edges and its $\O(\tau)$-separator tree $\ct$ with height $\eta = O(\log m)$, there is a deterministic data structure \textsc{DynamicSC} which maintains the edge weights $\vw$ from the IPM, and at every node $H \in \ct$, maintains two Laplacians $\ml^{(H)}$ and $\tsc(\ml^{(H)}, \partial H)$ dependent on $\vw$.	It supports the following procedures:
	\begin{itemize}
		\item $\textsc{Initialize}(G,\vw\in\R_{>0}^{m}, \epssc > 0)$:
		Given a graph $G$, initial weights $\vw$, projection matrix approximation accuracy $\epssc$, preprocess in $\widetilde{O}(\epssc^{-2} m)$ time.
		\item $\textsc{Reweight}(\vw\in\R_{>0}^{m}, \text{given implicitly as a set of changed coordinates})$: 
		Update the weights to $\vw$ in $\widetilde{O}( \epssc^{-2} \sum_{H\in\mathcal{H}}|F_H\cup\partial H| )$ time, where $\mathcal{H} \defeq \{H\in \ct:(\vw-\vw^\prev)|_{E(H)} \neq \mzero\}$ is the set of nodes containing an edge with updated weight. (Note that $\collN$ is a union of paths from leaf nodes to the root.)
		
		\item Access to Laplacian $\ml^{(H)}$ at any node $H \in \ct$ in time $\widetilde{O}\left( \epssc^{-2} |\partial H \cup F_H|\right)$.
		\item Access to Laplacian $\tsc(\ml^{(H)}, \bdry{\region})$ at any node $H \in \ct$ in time $\widetilde{O}\left( \epssc^{-2} |\partial H|\right)$.
	\end{itemize}
	Furthermore, with high probability, for any node $H$ in $\ct$, if $H$ is at level $i$, then
	\begin{align*} 
		 \ml^{(\region)} &\approx_{i \epssc} \sc(\ml[H], \bdry{H}\cup \elim{H}), \qquad \text{and} \\
		 \tsc(\ml^{(H)}, \bdry{\region}) &\approx_{\epssc} \sc(\ml^{(H)}, \bdry{\region}).
	\end{align*}
\end{theorem}

	\section{Solution maintenance}

Assuming the correct maintenance of Laplacians and Schur complements along a recursive separator tree, \cite{planarFlow} gave detailed data structures for maintaining the exact and approximate flow and slack solutions throughout the IPM.
Recall that the leaf nodes of the separator tree $\ct$ form a partition of the edges of $G$. 
In \cite{planarFlow}, each leaf node of the separator tree contains $O(1)$-many edges.
while in our case, each leaf node contains $O(\tau)$-many edges,
and therefore we update the approximate solution in a block manner. 
The data structures in \cite{planarFlow} naturally generalized from coordinate-wise updates to the block-wise case,
so we use their implementation in a black-box manner.
Here, we state a combined version of their main theorems.

We use $\vf_{[i]}$ to denote the subvector of $\vf$ indexed by edges in the $i$-th leaf node of $\ct$, and similarly $\vs_{[i]}$. 
We use $\vf_{[i]}^{(k)}$ and $\vs_{[i]}^{(k)}$ to denote the vector $\vf_{[i]}$ and $\vs_{[i]}$ at the $k$-th step of the IPM. 
Each block contains $O(\tau)$-many variables.

\begin{theorem}[{\cite[Theorem 9, 10]{planarFlow}}]\label{thm:sol-maintain}
	Given a graph $G$ with $m$ edges and its separator tree $\ct$ with height $\eta$,
	there is a randomized data structure 
	implicitly maintains the IPM solution pair $(\vf,\vs)$ undergoing IPM changes,
	and explicitly maintains its approximation $(\of,\os)$,
	and supports the following procedures with high probability against an adaptive adversary:
	\begin{itemize}
		\item $\textsc{Initialize}(G,\vf^{\init} \in\R^{m},\vs^{\init} \in\R^{m},\vv\in \R^{m}, \vw\in\R_{>0}^{m},\epsilon_{\mproj}>0,\overline{\epsilon}>0)$:
		Given a graph $G$, initial solutions $\vf^\init, \vs^{\init}$, initial direction $\vv$, initial weights $\vw$,
		target step accuracy $\epsilon_{\mproj}$ and target approximation accuracy
		$\overline{\epsilon}$, 
		preprocess in $\widetilde{O}(m \epsilon_{\mproj}^{-2})$ time, 
		set the implicit representations $\vf \leftarrow \vf^\init, \vs \leftarrow \vs^{\init}$,
		and set the approximations $\of \leftarrow \vf, \os \leftarrow \vs$.
		
		\item $\textsc{Reweight}(\vw\in\R_{>0}^{m},$ given implicitly as a set of changed weights): 
		Set the current weights to $\vw$ in $\widetilde{O}(\epsilon_{\mproj}^{-2} \tau K)$\footnote{The original bound here is $\widetilde{O}(\epsilon_{\mproj}^{-2}\sqrt{mK})$ where the $\sqrt{mK}$ factor comes from the fact that they can bound $\sum_{H\in \mathcal{H}}|F_H \cup \partial H|\leq \sqrt{mK}$ and  $\mathcal{H} = \{H\in\ct : (\vw-\vw^\prev)|_{E(H)}\neq \vzero\}$. See \cite[Section 9]{planarFlow} for more details.
        Here, we replace it by $\O(\tau K)$ using \cref{rmk:boundary-size}.}
		time, where where $K$ is the number of blocks changed in $\vw$.

		\item $\textsc{Move}(\alpha\in\mathbb{R},\vv\in\R^{m} $ given implicitly as a set of changed coordinates): 
		Implicitly update 
		\begin{align*} 
			\vs &\leftarrow \vs+\alpha\mw^{-1/2}\widetilde{\mproj}_{\vw} \vv, \\
			\vf &\leftarrow \vf+ \alpha \mw^{1/2}\vv - \alpha \mw^{1/2} \widetilde{\mproj}'_{\vw} \vv,
		\end{align*} 
		for some $\widetilde{\mproj}_{\vw}$ satisfying $\|(\widetilde{\mproj}_{\vw} -\mproj_{\vw}) \vv \|_2 \leq\eta\epssc \norm{\vv}_2$ and $\widetilde{\mproj}_{\vw} \vv \in \range{\mb}$,
		and some other $\widetilde{\mproj}'_{\vw}$ satisfying
		$\|\widetilde{\mproj}'_{\vw} \vv - \mproj_{\vw} \vv \|_2 \leq O(\eta \epssc) \norm{\vv}_2$ and
		$\mb^\top \mw^{1/2}\widetilde{\mproj}'_{\vw}\vv= \mb^\top \mw^{1/2} \vv$.
		
		The total runtime is $\widetilde{O}(\epsilon_{\mproj}^{-2}\tau K)$, where $K$ is the number of blocks changed in $\vv$.
		
		\item $\textsc{Approximate}()\rightarrow\R^{2m}$: Return the vector pair $\of,\os$ implicitly as a set of changed coordinates, 
		satisfying $\|\mw^{-1/2}(\of-\vf)\|_{\infty}\leq\overline{\epsilon}$ and $\|\mw^{1/2}(\os-\vs)\|_{\infty}\leq\overline{\epsilon}$, for the current weight $\vw$ and the current solutions $\vf, \vs$. 
		
		\item $\textsc{Exact}()\rightarrow\R^{2m}$: 
		Output the current vector $\vf,\vs$ in $\O(m \epssc^{-2})$ time.
	\end{itemize}
	Suppose $\alpha\|\vv\|_{2}\leq\beta$ for some $\beta$ for all calls to \textsc{Move}.
	Suppose in each step, \textsc{Reweight}, \textsc{Move} and \textsc{Approximate} are called in order. 
	Let $K$ denote the total number of blocks changed in $\vv$ and $\vw$ between the $(k-1)$-th and $k$-th \textsc{Reweight} and \textsc{Move} calls. Then at the $k$-th \textsc{Approximate} call,
	\begin{itemize}
		\item the data structure 
		sets $\of_{[i]} \leftarrow \vf^{(k)}_{[i]},\os_{[i]}\leftarrow \vs^{(k)}_{[i]}$
		for $O(N_k\defeq 2^{2\ell_{k}}(\frac{\beta}{\overline{\epsilon}})^{2}\log^{2}m)$ blocks $i$,
		where $\ell_{k}$ is the largest integer
		$\ell$ with $k \equiv 0\mod2^{\ell}$ when $k\neq 0$ and $\ell_0=0$, and
		\item the amortized time for the $k$-th \textsc{Approximate} call
		is $\widetilde{O}(\epsilon_{\mproj}^{-2}{\tau(K+N_{k-2^{\ell_k}})})$.
	\end{itemize}
\end{theorem}

We note that running the data structure above for each $2^{\ell}$ step, takes $\widetilde{O}(2^{2\ell}\tau)$ total time. For $2^{\ell}=\Omega(\sqrt{m})$  total steps, this is $O(m\tau)$. However, the initialization only takes $\widetilde{O}(m)$ time. Therefore, we restart the data structure every $2^{\ell}=\Omega(\sqrt{m/\tau})$ steps.

	\section{Proof of main theorems}

\begin{algorithm}

    \caption{Implementation of Robust Interior Point Method\label{algo:IPM_impl}}
    
    \begin{algorithmic}[1]
    
    \Procedure{$\textsc{CenteringImpl}$}{$\mb,\phi,\vf,\vs,t_{\mathrm{start}},t_{\mathrm{end}},\tau$}
    
    \State $G$: graph on $n$ vertices and $m$ edges with incidence matrix $\mb$
    \State $\mathsf{Solution}$: data structures for slack and flow maintenance 
        \Comment \cref{thm:sol-maintain}
    
    \State $\alpha \defeq \frac{1}{2^{20}\lambda}, \lambda \defeq 64\log(256m^{2})$
    
    \State $t\leftarrow t_{\mathrm{start}}$, $\of\leftarrow\vf,\os\leftarrow\vs,\ot\leftarrow t$, $\mw\leftarrow\nabla^{2}\phi(\of)^{-1},k\leftarrow 0$
    \Comment variable initialization
    \State $\vv_{i} \leftarrow \sinh(\lambda\gamma^{\ot}(\of,\os)_{i})\mu_i^{\ot}(\of_i,\os_i)$ for all $i \in [n]$ 
    \Comment data structure initialization
    \State $\mathsf{solution}.\textsc{Initalize}(G,\vf,\ot^{-1}\vs,\vv, \mw,\epssc=O(\alpha/\log m),\overline{\epsilon}=\alpha)$   

    \Comment choose $\epssc \leq \alpha<\eta$ in \cref{thm:sol-maintain} 
    
    \While{$t\geq t_{\mathrm{end}}$}
    
    \State $t\leftarrow\max\{(1-\frac{\alpha}{\sqrt{m}})t, t_{\mathrm{end}}\},k\leftarrow k+1$

    \State Update $h=-\alpha/\|\cosh(\lambda\gamma^{\ot}(\of,\os))\|_{2}$
    
    \State Update the diagonal weight matrix $\mw=\nabla^{2}\phi(\of)^{-1}$\label{line:step_given_3_impl}
    
    \State Update step direction $\vv_{i} \leftarrow \sinh(\lambda\gamma^{\ot}(\of,\os)_{i}) \mu_i^{\ot}(\of_i,\os_i)$ for all $i$ where $\of_{i}$ or $\os_i$ has changed \label{line:step_given_end_impl}
    
    \State $\mathsf{solution}.\textsc{Reweight}(\mw)$ \Comment update data structure with new weights
    \State $\mathsf{solution}.\textsc{Move}(h,\vv)$
    \Comment  update $\vf$ and $\vs$ 
    
    \State $\of,\ot\os\leftarrow\mathsf{solution}.\textsc{Approximate}()$ 
    \Comment maintain $\of,\os$  
    
    \label{line:step_user_end_impl}
    
    \If {$|\ot-t|\geq\alpha\ot$ \textbf{or} $k>\sqrt{m/\tau}$} \Comment{restart data structure}
        \State $\vf,\vs\leftarrow \mathsf{solution}.\textsc{Exact}()$
        \State $\ot\leftarrow t,k\leftarrow 0$ 
        \State $\mathsf{solution}.\textsc{Initalize}(G,\vf,\ot^{-1}\vs,\vv, \mw,\epssc=O(\alpha/\log m),\overline{\epsilon}=\alpha)$   
    \EndIf 
    
    \EndWhile
    
    \State \Return $\mathsf{solution}.\textsc{Exact}()$
    
    \EndProcedure
    
    \end{algorithmic}
    
    \end{algorithm}

\MainThm*

\begin{proof}
Since \cref{thm:IPM} requires an interior point in the polytope, we note one can find an interior point with $r\geq \frac{1}{4m}$ in $O(m)$ time, see \cite{planarFlow}.

Now, we bound the parameters $L,R,r$ in \cref{thm:IPM}. Clearly,
$L=\|\vc^{\mathrm{new}}\|_{2}=O(Mm)$ and $R=\|\vu^{\mathrm{new}}-\vl^{\mathrm{new}}\|_{2}=O(Mm)$.

The RIPM in \cref{thm:IPM} runs the subroutine \textsc{Centering} twice.
In the first run, the constraint matrix is the incidence matrix of a new underlying graph,
constructed by making three copies of each edge in the original graph $G$. 
Since copying edges does not affect treewidth, and our data structures allow for duplicate edges,
we use the implementation given in \textsc{CenteringImpl} (\cref{algo:IPM_impl}) for both runs.

By the guarantees of \cref{thm:sol-maintain},
we correctly maintain $\vf$ and $\vs$ at every step in \textsc{CenteringImpl}, 
and the requirements on $\of$ and $\os$ for the RIPM are satisfied.
Hence, \cref{thm:IPM} shows that we can find a circulation $\vf$
such that $(\vc^{\mathrm{new}})^{\top}\vf\leq\mathrm{OPT}-\frac{1}{2}$
by setting $\epsilon=\frac{1}{CM^{2}m^{2}}$ for some large constant
$C$ in \cref{alg:IPM_centering}. 
Note that $\vf$, when restricted to the original graph, is almost
a flow routing the required demand with flow value off by at most $\frac{1}{2nM}$. This
is because sending extra $k$ units of fractional flow from $s$ to
$t$ gives extra negative cost $\leq-knM$. Now we can round $\vf$
to an integral flow $\vf^{\mathrm{int}}$ with same or better flow
value using no more than $\widetilde{O}(m)$ time \cite{kang2015flow}.
Since $\vf^{\mathrm{int}}$ is integral with flow value at least the total demand
minus $\frac{1}{2}$, $\vf^{\mathrm{int}}$ routes the demand completely.
Again, since $\vf^{\mathrm{int}}$ is integral with cost at most $\mathrm{OPT}-\frac{1}{2}$,
$\vf^{\mathrm{int}}$ must have the minimum cost.

Finally, we bound the runtime of \textsc{CenteringImpl}. Before, we use the data structure for flow and slack maintenance, we need to construct the separator tree, this can be done in $\O(n\tau)$ time using~\cref{thm:construct-separator-tree}.
We initialize the data structures for flow and slack by \textsc{Initialize}.
Here, the data structures are given the first IPM step direction $\vv$ for preprocessing; the actual step is taken in the first iteration of the main while-loop.
At each step of \textsc{CenteringImpl}, we perform the implicit update of $\vf$ and $\vs$ using $\textsc{Move}$;  we update $\mw$ in the data structures using $\textsc{Reweight}$; and we construct the explicit approximations 
$\of$ and $\os$ using $\textsc{Approximate}$; 
each in the respective flow and slack data structures. 
We return the true $(\vf,\vs)$ by $\textsc{Exact}$. 
The total cost of \textsc{CenteringImpl} is dominated by $\textsc{Move}$, $\textsc{Reweight}$, and $\textsc{Approximate}$. 

Since we call \textsc{Move}, \textsc{Reweight} and \textsc{Approximate} in order in each step and the runtime for \textsc{Move}, \textsc{Reweight} are both dominated by the runtime for \textsc{Approximate}, it suffices to bound the runtime for \textsc{Approximate} only.  
\cref{thm:IPM} guarantees that there are $T=O(\sqrt{m}\log n\log(nM))$ total \textsc{Approximate} calls. We implement this by restarting the data structure for every $\sqrt{m/\tau}$ steps.

We note that that at the $k$-th step, the number of blocks changed in $\vw$ and $\vv$ is bounded by $K \defeq O(2^{2\ell_{k-1}} \log^2 m)$, where $\ell_k$ is the largest integer $\ell$ with $k \equiv 0 \mod 2^{\ell}$, or equivalently, the number of trailing zeros in the binary representation of $k$.
\cref{thm:IPM} further guarantees we can apply \cref{thm:sol-maintain} with parameter $\beta = O(1/\log m)$,
which in turn shows the amortized time for the $k$-th call is
\[
	\O(\epssc^{-2} {\tau(K + N_{k-2^{\ell_k}})}).
\]

where $N_{k} \defeq 2^{2\ell_k} (\beta/\alpha)^2 \log^2 m = O(2^{2\ell_k} \log^2 m)$,
where $\alpha = O(1/\log m)$ and $\epsilon_{\mproj}= O(1/\log m)$ are defined in \textsc{CenteringImpl}.
Observe that $K + N_{k-2^{\ell_k}} = O(N_{k-2^{\ell_k}})$. Now, summing over $T=\sqrt{m/\tau}$ steps, the time is
\begin{align*}
	O(\sqrt{m/\tau})  \sum_{k=1}^T \O(\tau{N_{k-2^{\ell_k}}}) &= 
	O(\sqrt{m/\tau})  \sum_{\ell=0}^{\log T} \frac{T}{2^\ell}\cdot \widetilde{O}(2^{2\ell}\tau) =\widetilde{O}(m).
\end{align*}

We note the initialization time is also $\O(m)$. Recall that \cref{thm:IPM}'s guarantee of total number of iterations is $O(\sqrt{m}\log n\log(nM))$, the data structure restarts for $\O(\tau)$-many times in total. The total runtime for the RIPM data structure is $\O (m\sqrt{\tau} \log M)$.

Hence, we conclude the overall running time is $\O(m\sqrt{\tau}\log M + n\tau)$.
\end{proof}

\ApproxTW*

Our algorithm requires some tree decomposition of the graph as input. We use the following lemma to construct the initial tree decomposition.
\begin{lemma}[\cite{brandt2019approximating}]\label{lem:bw19}
    For any $\tfrac23<\alpha<1$ and $0<\epsilon<1-\alpha$, given a graph $G$ with $n$ vertices and $m$ edges, if the graph $G$ contains an $\alpha$-balanced vertex separator of size $K$, then there is a randomized algorithm that finds a balanced vertex separator of size $\O(K^2/\epsilon)$ in $\O(mK^3/\epsilon)$ expected time. The algorithm does not require knowledge of $K$.
\end{lemma}

The following lemma establishes the relationship between max flow and balanced edge separators. 
We first give the relevant definitions.
For a given constant $c \leq 1/2$, a directed edge-cut $(S, \overline{S})$ is called a \emph{$c$-balanced edge separator} if both $|S| \geq cn$ and $|\overline{S}| \geq cn$. 
The capacity of the cut $(S, \overline{S})$ is the total capacity of all edges crossing the cut. The minimum $c$-balanced edge separator problem is the $c$-balanced edge separator with minimum capacity. A $\lambda$ pseudo-approximation to the minimum $c$-balanced edge separator is a $c'$-balanced cut $(S, \overline{S})$ for some other constant $c'$, whose capacities is within a factor of $\lambda$ of that of the minimum $c$-balanced edge separator.

\begin{lemma}[\cite{arora2016combinatorial}]
    An $O(\log n)$ pseudo-approximation to the minimum c-balanced edge separator in directed graphs can be computed using $\polylog n$ single-commodity flow computations on the same graph.
\end{lemma}

\begin{proof}[Proof of \cref{cor:tw}]
    It's well known that given a $O(\log n)$ approximation algorithm for finding a balanced vertex separator, one can construct a tree decomposition of width $O(\tw(G)\log n)$.  Specifically, the algorithm of~\cite{bodlaender1995approximating} find a tree decomposition by recursively using a balanced vertex separator algorithm.  

    Now, it suffices to show we can find a $\log(n)$ pseudo-approximation balanced vertex separator in $\O(m\cdot \tw(G)^3)$ expected time. Using the reduction from~\cite{leighton1999multicommodity}, we reduce the balanced vertex separator to directed edge separator on graph $G^*=(V^*,E^*)$, where
    \[
		V^*=\big\{v \mid v \in V\big\} \cup \big\{v^{\prime} \mid v \in V\big\},
	\]
and
	\[
	E^*= \big\{ \left(v, v^{\prime}\right) \mid v \in V\big\} \cup 
	\big\{\left(u^{\prime}, v\right) \mid(u, v) \in E \big\} \cup 
	\big\{\left(v^{\prime}, u\right) \mid(u, v) \in E\big\}.
\]

We note that $\tw(G^*) = O(\tw(G))$. 
This shows $G^*$ has a $2/3$-balanced vertex separator of size $O(\tw(G))$. We first use~\cref{lem:bw19} to construct a $\O(\tw(G)^2)$-separator tree for $G^*$. Then, we use the algorithm in~\cite{arora2016combinatorial} combined with our flow algorithm to find a balanced edge separator in $\O(m\cdot \tw(G))$ expected time. 
Hence, we can find a balanced vertex separator in $\O(m\cdot \tw(G)^3)$ expected time.
\end{proof}
	\bibliographystyle{alpha}
	\bibliography{ref-new}

\end{document}